\documentclass[pmlr]{jmlr}

\usepackage{longtable}

\usepackage{booktabs}
 \usepackage{siunitx}

\newcommand{\N}{\mathbb{N}}

\newcommand{\R}{\mathbb{R}}
\newcommand{\Z}{\mathbb{Z}}
\newcommand{\C}{\mathbb{C}}

\newcommand{\free}{\Sigma^*}

\newcommand{\mat}[1]{\mathbf{#1}}

\newcommand{\A}{\mat{A}}

\renewcommand{\H}{\mat{H}}

\newcommand{\balpha}{\boldsymbol{\alpha}}
\newcommand{\bbeta}{\boldsymbol{\beta}}

\newcommand{\Hfree}{\mathcal{H}^2(\free)}

\newcommand{\wfa}{\langle \balpha , \{\A_a\},  \bbeta \rangle}
\newcommand{\wa}{\langle \balpha , \A,  \bbeta \rangle}

\newcommand{\norm}[1]{\|#1\|}

\usepackage{color}
\usepackage{xcolor}

\theorembodyfont{\upshape}
\theoremheaderfont{\scshape}
\theorempostheader{:}
\theoremsep{\newline}

\jmlrvolume{1}
\jmlryear{2022}
\jmlrworkshop{LearnAut 2022}

\title[Towards Approximate Minimization in the Multi-Letter Case]{Towards an AAK Theory Approach to Approximate Minimization in the Multi-Letter Case}

\author{\Name{Clara Lacroce}\nametag{\thanks{Corresponding author}} \Email{clara.lacroce@mail.mcgill.ca}\\
\addr School of Computer Science, McGill University \& Mila , Montr\'eal, Canada
\AND
\Name{Prakash Panangaden} \Email{prakash@cs.mcgill.ca}\\
  \addr School of Computer Science, McGill University \& Mila , Montr\'eal, Canada
\AND
\Name{Guillaume Rabusseau} \Email{grabus@iro.umontreal.ca}\\
    \addr DIRO, Universit\'e de Montr\'eal \& Mila, Montr\'eal, Canada}

\begin{document}

\maketitle

\begin{abstract}
We study the approximate minimization problem of weighted finite automata (WFAs): given a WFA, we want to compute its optimal approximation when restricted to a given size. We reformulate the problem as a rank-minimization task in the spectral norm, and propose a framework to apply Adamyan-Arov-Krein (AAK) theory to the approximation problem. This approach has already been successfully applied to the case of WFAs and language modelling black boxes over one-letter alphabets \citep{AAK-WFA,AAK-RNN}. Extending the result to multi-letter alphabets requires solving the following two steps. First, we need to reformulate the approximation problem in terms of noncommutative Hankel operators and noncommutative functions, in order to apply results from multivariable operator theory. Secondly, to obtain the optimal approximation we need a version of noncommutative AAK theory that is constructive. In this paper, we successfully tackle the first step, while the second challenge remains open.
\end{abstract}
\begin{keywords}
Approximate minimization, Hankel matrices, AAK theory, weighted finite automata, language modelling.
\end{keywords}

\section{Introduction}
\label{sec:intro}

    The problem of minimizing an automaton has been well studied over the past seventy years. When dealing with quantitative models, like weighted or probabilistic automata, it becomes possible to define quantitative notions of model similarity, and to find \emph{approximately} minimal approximations. 
    In particular, given a minimal weighted finite automaton (WFA), the \emph{approximate minimization problem} consists in finding a WFA, smaller than the minimal one, that mimics its behaviour. We are interested in quantifying and minimizing the approximation error. 
    The approximate minimization problem is strictly related to knowledge distillation and extraction tasks \citep{WeissWFA19,Takamasa,eyraud2020,Ayache2018,Rabusseau19,WeissDFA18,merrill2022}. When the solution of the problem is optimal, this approach has a clear advantage compared to other methods, as it allows us to search for the best WFA among those of a predefined size. This is particularly useful when dealing with limited computing resources, or to improve interpretability. 
    
    Several norms can be considered to estimate the error. The approximate minimization problem was formalized by \citet{Balle19}, and the error measured with respect to the $\ell^2$ norm. In this paper, we reformulate the problem in terms of the WFAs' Hankel matrix $\H$ and look for a low-rank approximation in the spectral norm. This norm has the advantage that it can be used to compare different classes of models. Moreover, it is possible to find (and compute) a global minimum for the error in polynomial time \citep{AAK-WFA}. 
    In fact, the celebrated Adamyan-Arov-Krein (AAK) theory provides a way, based on properties of Hankel operators and complex functions, to find the optimal approximation of $\H$ within the class of Hankel matrices. We lay out a framework for the application of this theory to the approximate minimization problem, analyzing the case of one-letter and multi-letter alphabet separately. In the first case, standard AAK theory can be applied, and the proof of the AAK theorem tells us how to construct the optimal approximation. This setting has been studied by \citet{AAK-WFA} to obtain an algorithm, based on AAK theory, returning the optimal approximation of a class of WFAs. \citet{AAK-RNN} generalized this approach to find an (asymptotically) optimal approximation of a general black-box model trained for language modelling on sequential data, still under the one-letter assumption. Extending the work to the multi-letter case requires a noncommutative (NC) version of AAK theory. Tackling this problem is fundamental for the application of these results and to experimentally compare the performance of the spectral norm against other norms (for example behavioral metrics, word error rate, or normalized discounted cumulative gain). To achieve this, the following two challenges need to be addressed.
    \begin{itemize}
        \item[1.] To apply the AAK theorem it is necessary to reformulate the approximation problem in terms of NC Hankel operators, defined in an appropriate NC space.
        \item[2.]To find the optimal approximation, we need a constructive version of the AAK theorem.
    \end{itemize}
    In this paper, we tackle the first challenge and make the following contributions. We start by summarizing the approach used by \citet{AAK-WFA, AAK-RNN} in the one-letter case. Then, we reformulate the approximate minimization problem of models over multi-letter alphabets in terms of NC Hankel operators. Finally, we suggest a way to link the Hankel matrix of a WFA to a NC rational function.
    While the second challenge remains open, this constitutes a first, encouraging step towards its solution, since the rational function is key in the construction of the optimal approximation.

\section{Background}
    
    Let $\N$, $\Z$ and $\R$ be the sets of natural, integers and real numbers, respectively. 
    We use bold letters for vectors and matrices; all vectors considered are column vectors. 
    We denote with $\mat{v}(i)$, $\mat{M}(i,:)$ and $\mat{M}(:,j)$ the $i$-th component of the vector $\mat{v}$, and the $i$-th row and $j$-th column of $\mat{M}$, respectively. 
    Given $\mat{M}\in\R^{d_1\times d_2}$, $\mat{N}\in\R^{d_1'\times d_2'}$ we denote their \emph{Kronecker product} by $\mat{M}\otimes\mat{N}\in R^{d_1d_1'\times d_2d_2'}$ with entries given by $(\mat{M}\otimes\mat{N})((i-1)d_1'+i',(j-1)d_2'+j')=\mat{M}(i,j)\mat{N}(i',j')$.
    Let $\mat{M} \in \R^{p \times q}$ of rank $n$, the compact \emph{singular value decomposition} SVD of $\mat{M}$ is the factorization $\mat{M}=\mat{U}\mat{D}\mat{V}^{\top}$, where $\mat{U}\in \R^{p\times n}$, $\mat{D}\in \R^{n\times n}$, $\mat{V}\in \R^{q \times n}$ are such that $\mat{U}^{\top}\mat{U}=\mat{V}^{\top}\mat{V}=\mat{1}$, where $\mat{1}$ denotes the identity matrix, and $\mat{D}$ is a diagonal matrix.
    The columns of $\mat{U}$ and $\mat{V}$ are called left and right \emph{singular vectors}, while the diagonal entries of $\mat{D}$ are the \emph{singular values}. 
    The \emph{spectral radius} $\rho(\mat{M})$ of $\mat{M}$ is the largest modulus among its eigenvalues. 
    A \emph{Hilbert space} is a complete normed vector space where the norm arises from an inner product. 
    Let $\ell^2$ be the space of square-summable sequences over $\N$. Let $\mathbb{T}=\{z\in \C: |z|=1\}$ and $\mathbb{D}=\{z\in \C: |z|<1\}$ be the complex unit circle and the (open) complex unit disc, respectively. 
    Let $\mathcal{L}^p(\mathbb{T})$ be the space of measurable functions on $\mathbb{T}$ for which the $p$-th power of the absolute value is Lebesgue integrable. 

\subsection{Hankel matrix and Weighted Automata}

    Let $\Sigma$ be a fixed finite alphabet, $\Sigma^*$ the set of all finite strings with symbols in $\Sigma$, and $\varepsilon$ the empty string. Given $p,s \in \Sigma^*$, we denote with $ps$ their concatenation. Let $f : \Sigma^* \to \R$, we can consider a matrix $\H_f \in \R^{\Sigma^* \times \Sigma^*}$ having rows and columns indexed by strings and defined by $\H_f(p,s) = f(ps)$ for $p, s \in \Sigma^*$.
    \begin{definition}
         A (bi-infinite) matrix $\H \in \R^{\Sigma^* \times \Sigma^*}$ is \textbf{Hankel} if for all $p, p', s, s' \in \Sigma^*$ such that $p s = p' s'$, we have $\H(p,s) = \H(p',s')$. Given a Hankel matrix $\H \in \R^{\Sigma^* \times \Sigma^*}$, there exists a unique function $f : \Sigma^* \to \R$ such that $\H_f = \H$.
    \end{definition}
    A \textbf{weighted finite automaton} (WFA) of $n$ states over $\Sigma$ is a tuple $A = \wfa$, where $\balpha,$ $\bbeta \in \R^n$ are the vector of initial and final weights, respectively, and $\A_a \in \R^{n \times n}$ is the transition matrix associated with each symbol $a\in\Sigma$. In this paper, we only consider automata with real weights. In this case, every WFA $A$ realizes a function $f_A : \Sigma^* \to \R$, \emph{i.e.}, given a string  $x = x_1 \cdots x_t \in \Sigma^*$, it returns $f_A(x) = \balpha ^\top \A_{x_1} \cdots \A_{x_t} \bbeta = \balpha ^\top \A_x \bbeta$.
    Note that $f$ can be realized by a WFA if and only if $\H_f$ has finite rank $n$, in which case $n$ is the minimal number of states of any WFA realizing $f$ \citep{CP71,Fli}.
    


\subsection{Hankel Operators and AAK Theory}

    In this section, we introduce the definition of Hankel operators, both in the context of sequences and complex functions. Since we are concerned with functions on sequential data, it is easy to see why we might want to formalize Hankel operators in sequence spaces. The use of a functional representation, instead, is less intuitive, but it is needed to apply the results from AAK theory, which are stated in terms of complex functions.

    Given a function $f:\N \rightarrow \R$, we consider the Hankel matrix $\H_f$ defined by $\H_f(i,j)=f(i+j)$. This matrix can be interpreted as the expression of a linear \textbf{Hankel operator} $H_f:\ell^2 \rightarrow \ell^2$ in terms of the canonical basis of the sequence space. Alternatively, using the Fourier isomorphism, Hankel operators can be defined in a complex function space. In fact, the Hilbert space $\ell^2$ can be embedded into $\ell^2(\Z)$, which is isomorphic to $\mathcal{L}^2(\mathbb{T})$. Therefore, to each sequence $\boldsymbol{\mu}=(\mu_0, \mu_1, \dots)\in\ell^2$ we can associate two functions in the complex variable $z$:
    \begin{equation}
        \mu^-(z)=\sum_{j=0}^{\infty}\boldsymbol{\mu}_j z^{-j-1},  \quad\quad \mu^+(z)=\sum_{j=0}^{\infty} \boldsymbol{\mu}_j z^{j}.
    \end{equation}
    Conversely, we can associate any given function $\phi \in \mathcal{L}^2(\mathbb{T})$ with the sequence of its Fourier coefficients $\widehat{\phi}(n)$, for $n\in\Z$. The function space $\mathcal{L}^2(\mathbb{T})$ can be partitioned into two orthogonal subspaces, the \textbf{Hardy space} $\mathcal{H}^2$:
    \begin{equation}
        \mathcal{H}^2=\{ \phi \in \mathcal{L}^p(\mathbb{T}) : \widehat{\phi}(n)=0, n < 0\}
    \end{equation}
    and the negative Hardy space $\mathcal{H}^2_-$:
    \begin{equation}
        \mathcal{H}^2_-=\{ \phi \in \mathcal{L}^p(\mathbb{T}) : \widehat{\phi}(n)=0, n \geq 0\}
    \end{equation}
    \emph{i.e.} the spaces containing functions that have only nonnegative or negative Fourier coefficients, respectively. Note that $\mu^-(z)\in \mathcal{H}^2_-$ and $\mu^+(z)\in\mathcal{H}^2$, and it is possible to define the orthogonal projection
    \begin{equation}
        \mathbb{P}_-:\mathcal{L}^p(\mathbb{T})\rightarrow\mathcal{H}^2_-
    \end{equation}
    where $\mu^-(z)=\mathbb{P}_-(\mu^-(z)+\mu^+(z))$.
    We remark that it is possible to show that $\mathcal{H}^2$ is isomorphic to the set of square-integrable functions analytic on the disc. For a detailed presentation of these results we refer the reader to \citet{Nikolski}.
    \begin{definition}\label{Hankel2}
        Let $\phi$ be a function in the space $\mathcal{L}^2(\mathbb{T})$. A \textbf{Hankel operator} is an operator $H_{\phi}:\mathcal{H}^2 \rightarrow \mathcal{H}^2_-$ defined by 
        \begin{equation}
            H_{\phi}f=\mathbb{P}_-\phi f
        \end{equation}
        
        The function $\phi$ is said to be a \textbf{symbol} of $H_{\phi}$.
    \end{definition}
    We remark that the symbol is not unique, and that if $H_{\phi}$ is a bounded operator we can consider without loss of generality $\phi \in \mathcal{L}^{\infty}(\mathbb{T})$. In particular, the symbol of a finite rank Hankel operator is a rational function \citep{kronecker}. Moreover, the $\mathcal{L}^{\infty}$ norm of the symbol is related to the operator norm of the Hankel operator by the following relation: $\norm{H_{\phi}}\leq \norm{\phi}_{\infty}$ \citep{Nehari}. 
    
    Every Hankel matrix $\H$ satisfies the Hankel property $\H(j,k)=\{\alpha_{j+k}\}_{j,k \geq 0}$. Another way to express this property is to rephrase it as an operator identity. We consider the \textbf{shift operator} $S$ on the sequence space, with $S(x_0,x_1,\dots)=(0,x_0,x_1,\dots)$ and denote its left inverse by $S^*$. An operator $H$ is Hankel if and only if the following \textbf{Hankel equation} is satisfied:
    \begin{equation}\label{eq:hankel1}
        H S=S^*H.
    \end{equation}
    Alternatively, we can consider the shift operator $S$ in the function space, and generalize \equationref{eq:hankel1} for operators $H:\mathcal{H}^2 \rightarrow \mathcal{H}^2_-$:
    \begin{equation}\label{hankel_condition}
            HS=\mathbb{P}_-SH.
    \end{equation}
    
    We can now introduce the main result of \citet{AAK71}. 
    The theorem shows that for infinite dimensional Hankel matrices the constraint of preserving the Hankel property does not affect the approximation error.
    \begin{theorem}[AAK Theorem]\label{theorem:aakop}
        Let $H_{\phi}$ be a compact Hankel operator of rank $n$, and let $k<n$. We denote with $\H$ the matrix of $H_{\phi}$, and with $\sigma_0 \geq \dots \geq \sigma_{n-1}>0$ the singular numbers. Then there exists a unique Hankel operator $H_g$ with matrix $\mat{G}$ of rank $k$ such that: 
        \begin{equation}\label{eqoper}
            \norm{H_{\phi} - H_g} = \norm{\H-\mat{G}}= \sigma_k.
        \end{equation}
        We say that $\mat{G}$ is the optimal approximation of size $k$ of $\H$.
    \end{theorem}
    Note that the proof of this theorem is constructive, so it provides us with a way to build the optimal approximation. The method relies on $\phi$ and $g$, the symbols of the original operator and of the best approximation, respectively, and on the following fundamental inequality:
    \begin{equation}\label{eq:aakineq}
        \norm{H_{\phi} - H_g} \leq \norm{\phi-g}_{\infty}\leq \sigma_k.
    \end{equation}
    Therefore, if we want to apply AAK theory to our setting, it is fundamental to associate a complex function to the Hankel matrix we are trying to minimize. 

\subsection{Multivariable Operator Theory}\label{sec:nc}
    
    The ideal formalism to extend the previous results to a noncommutative setting is provided by the theory of Fock spaces. The choice of this kind of space is standard in noncommutative multivariable operator theory.
    
    Let $\mathcal{H}_n$ be a Hilbert space, $\mathcal{H}_n^{\otimes k}$ the tensor product of $k$ copies of $\mathcal{H}_n$, and $\mathcal{H}_n^{\otimes 0}:=\C$.
    \begin{definition}
        Let $\mathcal{H}_n$ be a $n$-dimensional Hilbert space. The (symmetric) \textbf{Fock space} $F^2$ of $\mathcal{H}_n$ is:
        \begin{equation}
            F^2=F^2(\mathcal{H}_n)=\bigoplus_{k\geq0}\mathcal{H}_n^{\otimes k}= \mathbb{C} \oplus \mathcal{H}_n \oplus  (\mathcal{H}_n \otimes \mathcal{H}_n) \oplus \dots
        \end{equation}
    \end{definition}
    Let $\free$ be the free monoid on $n$ generators $g_1,\dots g_n$, with identity element $g_0$. Given $\alpha\in\free$, with $\alpha=g_{i_1}g_{i_2}\cdots g_{i_k}$, we define its length by $|\alpha|=k$, and $|g_0|=0$. Analogously, we can define an element of the Fock space $e_{\alpha}=e_{i_1}\otimes e_{i_2}\otimes\cdots \otimes e_{i_k}$ and $e_{i_0}=1$. Note that $\mathcal{B}=\{e_{g_i}: g_i\in\free\}$ is an orthonormal basis for the Fock space $F^2$. The Fock space is isomorphic to the Hilbert space of square summable sequences indexed by elements in $\Sigma^*$.
    The Fock space can be also identified with $\mathcal{H}^2(\free)$, a canonical NC analogue of the Hardy space.
    Given a collection of $n$ NC variables (matrices or operators) $z=[z_1,\dots, z_n]$, with $z^{\alpha}:=z_{i_1}\cdot z_{i_2}\cdots z_{i_k}$, we can consider $f\in F^2$ and represent it as a formal power series: $f(z)=\sum_{\alpha\in\free}\widehat{f}_{\alpha}z^{\alpha}$, converging for $\sum_i||z_iz_i^*||<1$.
    We define the \textbf{NC Hardy space} as:
    \begin{equation}
        \Hfree=\left\{\sum_{\alpha\in\free}\widehat{f}_{\alpha}z^{\alpha}: \sum_{\alpha\in\free}||f_{\alpha}||^2<\infty \right\}.
    \end{equation}
    This means that we can choose between a ``sequence'' interpretation ($F^2$) or a ``functional'' interpretation ($\mathcal{H}^2(\free)$) of the results.
    We can now use sequences of operators to extend the definition of a Hankel operator in a way meaningful for NC spaces \citep{popescu}. This definition relies on extending the functional version of the Hankel equation (Equation \ref{hankel_condition}).
    \begin{definition}\label{def:ncop}
    Let $X=[X_1, \dots ,X_n]$, $X_i\in B(\mathcal{Y})$ be an arbitrary sequence of bounded operators on a Hilbert space $\mathcal{Y}$, and let $T=[T_1,\dots, T_n]$, $T_i\in B(\mathcal{H})$. Suppose $\mathcal{H}=\mathcal{H}_-\oplus\mathcal{H}_+$, with $\mathcal{H}_+$ invariant with respect to each $T_i\in B(\mathcal{H})$. Let $\mathbb{P}_-$ be the orthogonal projection on $\mathcal{H}_-$. A \textbf{NC Hankel operator} is a bounded linear operator $\Gamma:\mathcal{Y}\rightarrow \mathcal{H}_-$ such that:
        \begin{equation}
            \Gamma X_i= \mathbb{P}_-T_i\Gamma \quad\quad \text{for any}\,\,\, i=1,\dots,n.
        \end{equation}
    \end{definition}
   The definition of symbol provided in the commutative case can be generalized as follows. 
    \begin{definition}
        A \textbf{multiplier} is a bounded linear operator $A:\mathcal{Y}\rightarrow \mathcal{H}$ such that:
        \begin{equation}
            AX_i= T_iA \quad \text{for any}\,\, i=1,\dots,n.
        \end{equation}    
    \end{definition}
    Given a multiplier, it is always possible to associate with it a Hankel operator \citep{popescu} such that $||\Gamma_A||\leq ||A||$, and defined as:
    \begin{equation}
        \Gamma_Ay=\mathbb{P}_-Ay \quad \text{for}\,\, y\in\mathcal{Y}.
    \end{equation}
    
    We have the following noncommutative version of AAK theorem \citep{popescu}.
    \begin{theorem}[NC AAK Theorem]\label{theorem:fixedpoint2}
    Let $X=[X_1, \dots ,X_n]$, $X_i\in B(\mathcal{Y})$, and let $T=[T_1,\dots, T_n]$, $T_i\in B(\mathcal{H})$, be such that, for $y_i\in\mathcal{Y}$ and $h_i\in\mathcal{H}$:
    \begin{equation}
        \norm{X_1y_1+\dots+ X_ny_n}^2\geq \norm{y_1}^2+\dots+\norm{y_n}^2
    \end{equation}
    and 
    \begin{equation}
        \norm{T_1h_1+\dots+ T_nh_n}^2\leq \norm{h_1}^2+\dots+\norm{h_n}^2
    \end{equation}
    Let $\Gamma_A:\mathcal{Y}\rightarrow \mathcal{H}_-$ be a NC Hankel operator, with $\Gamma X_i= \mathbb{P}_-T_i\Gamma$ for any $i=1,\dots,n$. 
    Then
    there exists an optimal approximation of $\Gamma$ of size at most $k$.
    \end{theorem}
    
    We conclude this section with a quick overview of NC rational functions \citep{NCrational}. 
    A NC rational expression is any syntactically valid expressions involving several NC variables, scalars, $+$, $\cdot$, $^{-1}$ and parentheses. A \textbf{NC rational function} is an equivalence class between rational expressions, where we say that $r_1$ and $r_2$ belong to the same equivalence class if $r_1$ can be transformed into $r_2$ by algebraic manipulations. Unlike the commutative case, a NC rational function does not admit a canonical coprime fraction representation \citep{vinnikov2}. A ``canonical'' way to represent NC rational functions comes from the theory of formal languages \citep{Fli,berstel,SCHUTZENBERGER}. 
    In particular, every NC rational function containing $0$ in its domain admits a \textbf{minimal realization} of size $n$. If each $\mat{A}_j$ is a square matrix of size $n$, $\mat{b}$, $\mat{c}$ are vectors of size $n$ and each complex variable $z_j$ is a square matrix of size $m$, we have:
    \begin{equation}\label{eq:unpack}
        r(z)= \mat{c}^*\otimes\mat{1}_m\left(\mat{1}_n\otimes\mat{1}_m-\sum \mat{A}_j\otimes z_j\right)^{-1}\mat{b}\otimes\mat{1}_m=\mat{c}^*\left(\mat{1}-\sum \mat{A}_jz_j\right)^{-1}\mat{b}.
    \end{equation}
    The last equality can be used as a more concise representation of the rational function.

\section{A Framework for the Approximate Minimization Problem}
    
    Given a minimal WFA, we consider its Hankel matrix $\H$, with rank $n$ equal to the number of states. We propose to reformulate the approximation problem as a low-rank approximation of $\H$. Thus, we would like to find a matrix of rank $k<n$, approximating $\H$ in the spectral norm, and recover a WFA having $k$ states using the spectral method \citep{BalleCLQ14}. A well known theorem by \citet{Eckart} states that the optimal approximation of $\H$ is obtained by truncating its SVD, but the resulting matrix is not necessarily Hankel. This is a problem, as we want to extract from the matrix a WFA. Leveraging AAK theory, it is possible to find a Hankel matrix attaining the same bound as the optimal approximation. In the next two sections, we show how to associate to $\H$ a Hankel operator and a symbol in the case of one-letter and multi-letter alphabets. This is a necessary step in order to apply \theoremref{theorem:aakop}, since its constructive proof relies on the definition of a symbol. 

\subsection{One-letter Alphabets}
    
    Let $|\Sigma|=1$. Then, $\Sigma^*$ can be identified with $\N$ by associating to each string its length. 
    Since $\N$ can be embedded into $\Z$, we can interpret $f:\Sigma^*\rightarrow \R$ as $f:\Z \rightarrow \R$. This fundamental step allows us to apply the Fourier isomorphism to reformulate the problem in the Hardy space, where it can be solved using \theoremref{theorem:aakop}. This setting has been studied by \citet{AAK-WFA} and \citet{AAK-RNN} in the context of WFAs and black-box models, respectively.

\subsubsection{Defining a Hankel Operator and a symbol}
        
    Let $A=\wa$ be a minimal WFA with $n$ states over a one-letter alphabet, computing $f_A:\Sigma^* \rightarrow \R$ with Hankel matrix $\H$. To apply AAK theory, we need to associate with $\H$ a Hankel operator. We can define two different operators, $H_f$ and $H_{\phi}$. On the one hand, we can consider the Hankel operator acting over sequences $H_f:\ell^2\rightarrow\ell^2$, associated with the function $f_A:\Sigma^* \rightarrow \R$ with Hankel matrix defined by
    \begin{equation}
        \H(i,j)=f_A(i+j) \quad\quad\quad \text{for} \,\, i,j\geq 0.
    \end{equation}
    The matrix $\H$ can then be represented as:
    \begin{equation}
        \H=\H_f= \begin{pmatrix} f_A(0) & f_A(1) & f_A(2) & \dots \\
                              f_A(1) & f_A(2) & f_A(3) &\dots \\
                              f_A(2) & f_A(3) & f_A(4) &\dots \\
                                \vdots& \vdots &\vdots&\ddots
            \end{pmatrix}.
    \end{equation}
    On the other hand, we can interpret $\H$ as the matrix $\H_{\phi}$ associated with a Hankel operator over Hardy spaces, $H_{\phi}:\mathcal{H}^2 \rightarrow \mathcal{H}^2_-$. Now, the operator and matrix are related (by definition) to a complex function $\phi\in  \mathcal{L}^2(\mathbb{T})$, the symbol. The entries of the matrix are defined by means of the Fourier coefficients of $\phi$ as 
    \begin{equation}
        \H(j,k)= \widehat{\phi}(-j-k-1) \quad\quad\quad \text{for} \,\,j,k\geq 0.
    \end{equation}
    Note that the function $\mathbb{P}_-\phi=\widehat{\phi}(-j-k-1)$ is a complex rational function \citep{kronecker}. 
    The matrix $\H$ can be represented as:
    \begin{equation}
        \H=\H_{\phi}=\begin{pmatrix} \widehat{\phi}(-1) & \widehat{\phi}(-2) & \widehat{\phi}(-3) & \dots \\
                              \widehat{\phi}(-2) & \widehat{\phi}(-3) & \widehat{\phi}(-4) &\dots \\
                              \widehat{\phi}(-3) & \widehat{\phi}(-4) & \widehat{\phi}(-5) &\dots \\
                                \vdots& \vdots &\vdots&\ddots
            \end{pmatrix}.
    \end{equation}

    
    We can derive the relationship between $f_A$ and $\phi$: since we have $\H=\H_f=\H_{\phi}$, the two representations of the Hankel matrix need to coincide. We  have:
    \begin{equation}
        \H= \begin{pmatrix} f_A(0) & f_A(1) & f_A(2) & \dots \\
                              f_A(1) & f_A(2) & f_A(3) &\dots \\
                              f_A(2) & f_A(3) & f_A(4) &\dots \\
                                \vdots& \vdots &\vdots&\ddots
            \end{pmatrix}
            =\begin{pmatrix} \widehat{\phi}(-1) & \widehat{\phi}(-2) & \widehat{\phi}(-3) & \dots \\
                              \widehat{\phi}(-2) & \widehat{\phi}(-3) & \widehat{\phi}(-4) &\dots \\
                              \widehat{\phi}(-3) & \widehat{\phi}(-4) & \widehat{\phi}(-5) &\dots \\
                                \vdots& \vdots &\vdots&\ddots
            \end{pmatrix}
    \end{equation}
    from which we 
    obtain: 
    \begin{equation}\label{eq:symbol}
        f(n)=\widehat{\phi}(-n-1).    
    \end{equation}
    Therefore, in the case of a WFA $A=\wa$, we obtain the rational function:
    \begin{equation}\label{eq:symbolsum}
        \mathbb{P}_-\phi= \sum_{k\geq 0}f(k) z^{-k-1} = \sum_{k\geq 0} \balpha^{\top}\A^k \bbeta z^{-k-1} = \balpha^{\top}(z\mat{1}-\A)^{-1} \bbeta,
    \end{equation}
    where the last equality holds if $\rho(A)<1$. Now that we have the WFA's symbol, we have all the elements needed to find the best approximation using the constructive proof of \theoremref{theorem:aakop} \citep{AAK-WFA}.

\subsection{Multi-letter Alphabets}
    
    In this section, we consider a WFA over $\Sigma$, with $|\Sigma|=d>1$, so that $\Sigma^*$ is the free monoid generated by $d$ elements. $\Sigma^*$ is not abelian, so it cannot be embedded into $\Z$, and we cannot directly apply Fourier analysis like in the previous section. We first find a noncommutative version of \equationref{hankel_condition}, and suitable transformations to play the roles of the shifts. We then find an appropriate generalization of the Hardy spaces. This allows us to define an equivalent of \definitionref{def:ncop} in the case of Hankel matrices arising from a WFA. Finally, we associate the NC Hankel operator with a NC rational function by leveraging a property of the multipliers.
    

\subsubsection{Defining a Hankel Operator and a symbol}
    
    A WFA $A$ over $\Sigma$, with $|\Sigma|=d$, computes a function $f:\Sigma^*\rightarrow \R$, with Hankel matrix $\H$. This function can be interpreted as an element in the Fock space $F^2$ (see \appendixref{example:fock2}). We consider the shift operators defined on the Fock space. 
    For $i=1,\dots, d$, the \textbf{NC left shift} $S=(S_1,\dots, S_d)$ and \textbf{NC right shift} $R=(R_1,\dots, R_d)$ are defined by:
        \begin{equation}
             S_i (e_{\alpha}):= e_i\otimes e_{\alpha}=e_{i\alpha}, \quad \quad R_i (e_{\alpha}):=e_{\alpha}\otimes e_i = e_{\alpha i}.
        \end{equation}
    We can express the right shift in terms of the left one by using a unitary operator $U$, the \emph{flipping operator}: $ R_i=U^*S_iU$, where 
    \begin{equation}
        U(e_{i_1}\otimes e_{i_2}\otimes\cdots\otimes e_{i_k})=e_{i_k}\otimes \cdots\otimes e_{i_2}\otimes e_{i_1}.
    \end{equation}
    We obtain a NC version of \equationref{eq:hankel1}. 
    \begin{theorem}\label{thm:hankeleq}
        Let $|\Sigma|=d$, and let $\H$ be a WFA's Hankel matrix. Let $S$ and $R$
        be the NC left and right shifts on $F^2$, $S^*$ and $R^*$ their adjoints. Then, the following equation holds:
        \begin{equation}\label{eq:hankelnc}
            \H S_i=R^*_i\H \quad\quad\quad \text{for}\,\, i=1,\dots,d.
        \end{equation}
    \end{theorem}
     \begin{proof}
        For this proof, we leverage the functional representation of the Fock space. We recall that in the NC Hardy space, the left shift is equivalent to left multiplication by one of the noncommutative variables: $S_if=z_if$. Moreover, the function $f:\Sigma^*\rightarrow \R$ associated to the Hankel matrix can be represented by means of a formal power series in the NC Hardy space $f= \sum_{\alpha\in\Sigma^*}f(\alpha)z^{\alpha}$.
        This function corresponds to the first column of the Hankel matrix. Analogously, it is easy to see that the column at index $\alpha$ is: 
        \begin{equation}
            \H e_{\alpha}= \sum_{\beta\in\Sigma^*}f(\beta\alpha)z^{\beta}.
        \end{equation}
        Therefore:
        \begin{equation}
             \H S_i(e_{\alpha})= \sum_{\beta\in\Sigma^*}f(\beta i\alpha)z^{\beta}.
        \end{equation}
        On the other hand, we can consider the adjoint of the right shift:
        \begin{equation}
            R^*_i\H e_{\alpha}= R^*_i\sum_{\beta\in\Sigma^*}f(\beta\alpha)z^{\beta} =\sum_{\beta'\in\Sigma^*}f(\beta'i\alpha)z^{\beta'}.
        \end{equation}
        Thus, for any $i=1,\dots,d$, we have: $\H S_i=R^*_i\H$, which  concludes the proof. This shows that the Hankel matrix arising from a WFA defined over a multi-letter alphabets satisfies the NC version of the Hankel equation. 
    \end{proof}
    To extend \definitionref{def:ncop}, we need to find appropriate spaces $\mathcal{Y}$, $\mathcal{H}_-$, and $\mathcal{H}=\mathcal{H}_-\oplus \mathcal{H}_+$.
    It has become clear that the natural noncommutative generalization of $\ell^2(\N)$ is the Fock space, and that the NC Hardy space generalizes the Hardy space, so we set $\mathcal{Y}=\mathcal{H}_+=F^2$ (or $\mathcal{Y}=\mathcal{H}_+=\mathcal{H}^2(\Sigma^*)$). 
    In the one-letter case, the role of $\mathcal{H}$ was played by $\mathcal{L}^2(\mathbb{T})\cong \ell^2(\Z)$. A function $f\in\mathcal{L}^2(\mathbb{T})$ can be represented using the sequence of its Fourier coefficients, indexed by powers of the complex variable $z$.
    Analogously, we can set $\mathcal{H}=F_0^2 \oplus F^2$, and interpret it as the set of infinite sequences that are indexed by negative and nonnegative powers of the NC variables $z_1,\dots z_d$ (the $F_0^2$ and $F^2$ components, respectively). In \appendixref{example:fock2} we present an example of the application of this mathematical framework to a WFA over a $2$-letter alphabet.
    The following theorem shows that the formalization we chose is not only suitable to describe the Hankel matrix of a WFA, but it also leads to an appropriate definition of NC Hankel operator.
    \begin{theorem}\label{theorem:ncHank}
        Let $S=(S_1,\dots, S_d)$, $R=(R_1,\dots, R_d)$ be the left and right shifts on $F^2$, $S^*$ and $R^*$ their adjoints. Let $\mathcal{Y}=F^2$, $\mathcal{H}=F_0^2 \oplus F^2$, where $F^2_0=\bigoplus_{k>0}(\R^d)^{\otimes k}$. We set $\mathcal{H}_-=F^2_0$ and $\mathcal{H}_+=F^2$, and we define, for $i=1, \dots, d$, a bilateral shift on $\mathcal{H}$:
        \begin{equation}
            \begin{cases}
                \overline{R}_i(e_{\alpha})= R^*_i(e_{\alpha}) \quad\quad\quad \text{for} \,\,\, e_{\alpha}\in \mathcal{H}_- \\
                \overline{R}_i(e_{\alpha})= R_i(e_{\alpha}) \quad\quad\quad \text{for} \,\,\, e_{\alpha}\in \mathcal{H}_+
            \end{cases}.
        \end{equation}
        Let $\mathbb{P}_-$ be the orthogonal projection on $\mathcal{H}_-$.
        
        Then, the operator $H:\mathcal{Y}\rightarrow \mathcal{H}_-$ defined by the following property:
        \begin{equation}
            H S_i= \mathbb{P}_-\overline{R}_i H \quad\quad \text{for any}\,\,\, i=1,\dots,n
        \end{equation}
        is a NC Hankel operator according to \definitionref{def:ncop}. 
    \end{theorem}
     \begin{proof}
        In order to prove this theorem we need to verify that $\mathcal{H}_-\subset \mathcal{H}$, and that if $\mathcal{H}= \mathcal{H}_-\oplus\mathcal{H}_+$, then $\mathcal{H}_+$ is invariant under each $\overline{R}_i$. In particular, we want to show that these properties are satisfied when we set $\mathcal{Y}=F^2$, $\mathcal{H}_-=F^2_0$ and $\mathcal{H}=F_0^2 \oplus F^2$. The first property follows directly from the definition of $\mathcal{H}$: $F_0^2\subset F_0^2 \oplus F^2$. As for the second property, we note that $\mathcal{H}_-=F^2_0$, so it follows by definition that $\mathcal{H}_+=F^2$. 
            
        We want to show:
            \begin{equation}
                \overline{R}_1\mathcal{H}_++\dots + \overline{R}_d\mathcal{H}_+\subseteq \mathcal{H}_+,
            \end{equation}
        \emph{i.e.} that for any $h_i\in \mathcal{H}_+$ we have
        \begin{equation}
            \overline{R}_1h_1+\dots +\overline{R}_dh_d \in \mathcal{H}_+.
        \end{equation}
         Since $\overline{R}_i(e_{\alpha})= R_i(e_{\alpha})$ for $e_{\alpha}\in F^2$, the condition can be reformulated as:
         \begin{equation}
             R_i(e_{\alpha_1})+\dots + R_d(e_{\alpha_n}) \in F^2,
         \end{equation}
        which holds by definition of $R$, since $R_i(e_{\alpha})=e_{\alpha i} \in F^2$ for any $\alpha$, and the linear combination of elements in $F^2$ is an element in $F^2$.
    \end{proof}
    In the next theorem, we show that the properties needed in \theoremref{theorem:fixedpoint2} hold in our setting, \emph{i.e.} that NC AAK theory can be applied to the study of WFAs.
    \begin{theorem}\label{theorem:shifts}
        Let $H:\mathcal{Y}\rightarrow \mathcal{H}_-$, with $H S_i= \mathbb{P}_-\overline{R}_i H$, be the NC Hankel operator defined in the previous theorem. Then:
        \begin{itemize}
            \item[(a)] $\norm{S_1y_1+\dots + S_dy_d}^2\geq \norm{y_1}^2+\dots+\norm{y_d}^2$ for $y_i\in\mathcal{Y}$
            \item[(b)] $\norm{\overline{R}_1h_1+\dots + \overline{R}_dh_d}^2\leq \norm{h_1}^2+\dots+\norm{h_d}^2$ for $h_i\in \mathcal{H}$.
        \end{itemize}
    \end{theorem}
    \begin{proof}
    \begin{itemize}
        \item[(a)] 
        Leveraging the fact that the shifts have pairwise orthogonal ranges, so $S_i^*S_j=\mat{1}\delta_{i,j}$, and that each $S_i$ is an isometry, we obtain:
        \begin{align*}
            \norm{S_1y_1+\dots+S_dy_d}^2&=\langle S_1y_1,S_1y_1\rangle+\langle S_1y_1,S_2y_2 \rangle+\dots +\langle S_dy_d,S_dy_d\rangle \\
            &= \langle S_1y_1,S_1y_1\rangle+\langle S_2y_2,S_2y_2\rangle+\dots+\langle S_ny_n,S_ny_n\rangle\\
            &= \norm{y_1}^2+\dots+\norm{y_2}^2.
        \end{align*}
        \item[(b)] 
        The shifts have orthogonal ranges, so the result holds with the equality.
    \end{itemize}
    \end{proof}
    While the choice of the Fock space seems pretty natural, other spaces containing $F^2$ could have played the role of $\mathcal{H}$, for example the free group over $d$ elements. We show in \appendixref{apd:free} why this choice is not ideal in our setting.
    
    As seen in \sectionref{sec:nc}, in the noncommutative case a role similar to that of the symbol is played by an operator, the multiplier. We want a functional representation of the multiplier depending on the original Hankel matrix (like the symbol in the one-letter case). To achieve this, we first analyze the multiplier and find that, with minimal manipulations, we can get a functional description of it. Then, we show that this description is strictly related to the original Hankel operator, and can be used to rewrite \equationref{eq:aakineq} in the NC case. 
    
    We start by noting that, using the flipping operator,  we can rewrite the property of the multiplier as: 
    \begin{equation}
        UA S_a= S_a UA.
    \end{equation}
    The operators commuting with the left shift are called $S$-analytic operators, and can be represented using a function $\theta$. An  $S$-analytic operator $G$ has NC symbol $\theta$ if, for every $v$, $GS_a v = S_a\theta v$ \citep{popescuthesis,popescu_oponfock}.
    Concretely, this means that we can represent the operator $UA$ in terms of its NC symbol $\theta$, which corresponds to the multiplication by the first column of the matrix of $UA$ (this follows from \citet[Theorem 1.6]{popescuthesis}). Moreover, it is easy to show that $\norm{UA}=\norm{U\theta}_{\infty}$. Thus, if $H$ is a NC Hankel operator with multiplier $A$, and $\theta$ is the NC symbol of $UA$, we have:
    \begin{equation}
	    \norm{UH}\leq\norm{UA}=\norm{U\theta}_{\infty}.
    \end{equation}
    We refer to $U\theta$ as the \textbf{NC flipped symbol} of $H$. Note that it can be written as:
    \begin{equation}
        U\theta=\phi+c \quad\quad\quad \text{for} \,\, \phi\in H^2_0(\Sigma^*), \, c\in H^2(\Sigma^*)
    \end{equation}
    By construction, $\phi$ corresponds to the multiplication by the first column of $\H$.
    Note that if $R$ is a bounded operator, it is easy to show that
    \begin{equation}
        \norm{UH-R}=\norm{H-U^*R}.
    \end{equation}
    Moreover, since $UH$ is also a NC Hankel operator, it makes sense to search for its optimal approximation. If we denote with $UG$ the best approximation of $UH$, we have that $G$ is the best approximation of $H$. We have a NC generalization of \equationref{eq:aakineq}:
    \begin{equation}\label{eq:ncaak}
	    \norm{UH-UG}\leq\norm{UA-UB}\leq\norm{\phi+c-\psi-d}_{\infty}.
    \end{equation}
    
    We conclude by deriving an expression for the NC flipped symbol associated to a WFA. 
    Let $A=\wfa$ be a WFA computing a function $f$, let $\H$ be its Hankel matrix and $H$ the NC Hankel operator. The NC flipped symbol associated with $H$ is defined using the entries of the first column of $\H$. Its series expression is:
    \begin{equation}
        \mathbb{P}_-(\phi+c) =\sum_{a\in\free}f(a)z^{a} =\sum_{a\in\free}\balpha^{\top}\A^{a}\bbeta z^{\alpha}=\balpha^{\top}(\mat{1}-\sum \A_j z_j)^{-1}\bbeta.
    \end{equation}
    Note that $\phi$ is a rational function: in the noncommutative case also, there is a tight connection between WFAs and (NC) rational functions. This is very relevant, since in the one-letter case rewriting \equationref{eq:aakineq} in terms of the WFA's parameters is the key step to find the best approximation \citep{AAK-WFA}. At this stage, it is not clear if the proof of \theoremref{theorem:ncHank} can be made constructive. Nonetheless, by obtaining a noncommutative counterpart of this equation, expressed using the parameters of a WFA, we have built the machinery necessary to attack the problem in the case of multi-letter alphabets.

\section{Conclusion}
    
    In this paper, we study the approximate minimization problem of weighted finite automata with real weights. We 
    propose a way to associate a Hankel operator and a complex rational function to the Hankel matrix of a given WFA. This allows us to highlight the connections between approximate minimization and AAK theory. The application of AAK theory to the approximate minimization problem in the one-letter setting has been studied by \citet{AAK-WFA} and \citet{AAK-RNN} for WFAs and black boxes, respectively. To the best of our knowledge, this is the first attempt to apply AAK theory to automata in the multi-letter case setting.
    
    The approximate minimization problem is an interesting alternative to extraction when trying to approximate a black box (like RNNs) with a WFA. It allows to find the best approximation of a given size, directly improving interpretability and reducing the computational cost. The results in this paper can be easily generalized to the black-box setting. 
    In particular, the definition of the Hankel operator is the same. A little more care is needed to extend the computation of the symbol, given that the matrix of a black box does not necessarily have finite rank (so the corresponding function might not be rational).
    
    The framework we proposed is a key step towards solving the approximate minimization problem, as it allows us to rephrase it in terms of noncommutative AAK theory, where we know that a solution exists \citep{AAK71,popescu}. In the commutative setting, this is enough to construct the optimal approximation of a given size. Unfortunately, in the noncommutative setting AAK theorem is not constructive, so the problem of finding the best approximation remains open. Recent progress in the field of noncommutative multivariable operator theory \citep{blaschke,ball_bolotnikov_2021} leaves us hopeful that this challenge can be addressed. We think that the problem of constructing the optimal approximation is very relevant, as solving it would allow us to find a provable algorithm for the approximate minimization problem of black boxes, and provide us with a metric between different classes of models.

\acks{This research has been supported by NSERC Canada (C. Lacroce, P. Panangaden) and Canada CIFAR AI chairs program (G. Rabusseau). The authors would like to thank Doina Precup for supporting this work, Borja Balle for fruitful discussions and inputs, and Maxime Wabartha for feedback on the structure of the paper.
 }

\newpage

\bibliography{bibliography}

\newpage 
\appendix

\section{Example} \label{example:fock2}

     \begin{example}
        Let $\Sigma=\{a,b\}$, $\varepsilon$ the empty string. $\Sigma^*$ corresponds to the free monoid generated by two elements, where the generators are $g_1=a$ and $g_2=b$. A word $\alpha=aba$ can be seen as an element in $\Sigma^*$, with $\alpha=aba=g_1g_2g_1$, and the corresponding element in the Fock space $F^2$ is $e_{\alpha}=e_1\otimes e_2\otimes e_1$. A function $f:\Sigma^*\rightarrow \R$ can be viewed either as an element in the Fock space $F^2$, using a sequence interpretation:
        \begin{equation}
            (f(\varepsilon),f(a),f(b),f(aa),f(ab),f(ba),f(bb),f(aaa),\dots)\in F^2=\bigoplus_{k\geq0}(\R^2)^{\otimes k},
        \end{equation}
        or as a power series in the NC Hardy space  $\mathcal{H}^2(\Sigma^*)$, using a functional interpretation:
        \begin{equation}
             f(\varepsilon)+f(a)z_1+f(b)z_2+f(aa)z_1^2+f(ab)z_1z_2+f(ba)z_2z_1+ \dots=\sum_{\alpha\in\Sigma^*}f(\alpha)z^{\alpha}.
        \end{equation}
        As we can see, we obtain a bi-infinite sequence, indexed by the powers of the NC variables:
        \begin{table}[htbp]
        \begin{center}
        \begin{tabular}{c c c c c c c c c c c}
            ($\dots,$ & $f(a^{-2}),$ & $f(b^{-1}),$ & $f(a^{-1}),$ & $f(\varepsilon),$ & $f(a),$ & $f(b),$ & $f(aa),$ & $f(ab),$ &$\dots$) \\
            $\dots$ & $z_1^{-2}$ &  $z_2^{-1}$  & $z_1^{-1}$ &  $z_1^0z_2^0$ & $z_1^1$ & $z_2^1$ & $z_1^{2}$ & $z_1^{1}z_2^1$ & $\dots$
        \end{tabular}
        \end{center}
    \end{table}
        
    Now, we can consider the right shift $S=(S_1,S_2)$, with: 
        \begin{equation}
             S_1 (e_{\alpha})= e_{a\alpha}, \quad  S_2 (e_{\alpha})= e_{b\alpha}.
        \end{equation}
    The adjoint of $S$ is defined as:
        \begin{equation}
            S^*_1(e_{\alpha})= \begin{cases}
                e_{\alpha'} \quad\quad\quad \text{if} \,\,\, \alpha=a\alpha'\\
                0 \quad\quad\quad \text{otherwise}.
            \end{cases}
        \end{equation} 
    The right shift and its adjoint can be defined in a similar way.
    
    Let $A=\wfa$ be a WFA computing a function $f$, with Hankel matrix $\H$.
    \begin{equation}
         \H=\begin{pmatrix}  f(\varepsilon) & f(a) & \dots & f(ba)& \dots & f(aba) &  \dots\\
        f(a) & f(aa)  &\dots & f(aba) &  \dots&f(aaba)&\dots\\
                              f(b)  &f(ba) &\dots& f(bba) &  \dots&f(baba)&\dots\\
                              f(aa)  &f(aaa) &\dots& f(aaba) &  \dots&f(aaaba)&\dots\\
                              f(ab)  &f(aba) &\dots& f(abba) &  \dots&f(ababa)&\dots\\
                              f(ba)  &f(baa) &\dots& f(baba) &  \dots&f(baaba)&\dots\\
                              \dots & \dots &\dots &\dots &\dots &\dots &\dots 
            \end{pmatrix}.
    \end{equation}
    It is easy to see that:
    \begin{equation}
            \H S_a (e_{ba}) = \H e_{aba}=\sum_{\beta\in\Sigma^*}f(\beta aba)z^{\beta}=f(aba)z_0+ f(aaba)z_1+f(baba)z_2+\dots
    \end{equation}
    On the other hand, if we consider the adjoint of the right shift, we have:
    \begin{equation}
          R_a^* \H (e_{ba}) = R^*_a \sum_{\beta\in\Sigma^*}f(\beta ba)z^{\beta}=\sum_{\beta'\in\Sigma^*}f(\beta'aba)z^{\beta'} = f(aba)z_0+ f(aaba)z_1+\dots.
    \end{equation}
    We can obtain the same results for $S_b$ and $R^*_b$, so we can see that the Hankel equation holds.
    \end{example}

\section{The Free Group}\label{apd:free}

    We denote with $\mathbb{F}^*_d$ the free group on $d$ elements, and with $\ell(\mathbb{F}^*_d)$ the set of sequences indexed by elements in the free group. We can now show that, by setting $\mathcal{H}=\ell(\mathbb{F}^*_d)$, the conditions of Theorem \ref{theorem:ncHank} are satisfied, but the ones of Theorem \ref{theorem:shifts} are not. It is easy to see that $\mathcal{H}_-\subset \mathcal{H}$, $\mathcal{H}_+$ is invariant under the bilateral shift, and that property $(a)$ of Theorem \ref{theorem:shifts} is satisfied. On the other hand, property $(b)$ does not hold anymore. The components of the bilateral shifts don't have orthogonal ranges, as $\overline{R}_i^*\overline{R}_j\in \mathcal{H}$ even when  $i\neq j$. Intuitively the space is ``too big'' for property $(b)$ to hold. If we consider the intuition provided earlier about indexing the elements of $\mathcal{H}=F_0^2\oplus F^2$ using negative and nonnegative exponents, we have that in the case of the free group any combination of positive and negative exponents is allowed. Therefore, when defined on $\ell(\mathbb{F}^*_d)$, the adjoint of the shift is:
    \begin{equation}
            \overline{R}^*_i(e_{\alpha})= \begin{cases}
                e_{\alpha'} \quad\quad\quad \text{if} \,\,\, \alpha=\alpha'i\\
                e_{\alpha i^{-1}} \quad\quad\quad \text{otherwise}.
            \end{cases}
    \end{equation}
    Ultimately, we want to apply \theoremref{theorem:fixedpoint2}, in order to conclude that it is possible to find a optimal approximation of the NC Hankel operator associated to a WFA. For this to happen, we need \theoremref{theorem:shifts} to hold. Thus, the free group is not a viable option in our setting.
\end{document}